\documentclass{ifacconf}

\usepackage{graphicx}      
\usepackage{natbib}        
\usepackage{amsmath}
\usepackage{amsfonts}
\let\theoremstyle\relax
\usepackage{amsthm}
\usepackage{booktabs}
\usepackage{multirow}
\usepackage{tabularx}
\usepackage{mathtools}
\newtheorem{lemma}{Lemma}
\newtheorem{assumption}{Assumption}
\DeclareMathOperator*{\argmin}{\textnormal{argmin}}
\newcommand{\hop}{\mathsf{H}}
\usepackage{etoolbox}
\DeclareMathOperator*{\argmax}{argmax}

\AtBeginEnvironment{proof}{\vspace{-0.8em}}
\AtEndEnvironment{proof}{\vspace{-0.6em}}

\usepackage{eso-pic}
\AddToShipoutPictureBG*{%
\AtPageUpperLeft{%
\setlength\unitlength{1in}%
\hspace*{\dimexpr0.5\paperwidth\relax}
\makebox(0,-1.75)[c]{
\begin{tabular}{c c}
Enrico Dozzi \emph{et al.},
D-Optimized Sampling Design for System Identification, \\
To appear in 23rd IFAC World Congress, Busan, Corea, 2026
uploaded to ArXiv \today \\
\end{tabular}}}}

\begin{document}
\begin{frontmatter}

\title{D-Optimized Sampling Design for System Identification} 

\author[First]{Enrico Dozzi},
\author[First]{Tom Oomen}, 
\author[First]{Rodrigo A. González}

\address[First]{Department of Mechanical Engineering, Eindhoven University of Technology, Eindhoven, The Netherlands.}

\begin{abstract}                
Traditional system identification with multisine inputs relies on uniform sampling and periodic excitation to preserve Fourier orthogonality and avoid spectral leakage, limiting its use in scenarios with irregular sampling or nonperiodic inputs.
This work investigates continuous-time system identification under nonperiodic multisine excitation and nonuniform sampling.
We develop a nonparametric frequency response function estimator suited to such conditions and design irregular sampling schemes that enhance the informativeness of measurements and reduce spectral leakage.
The proposed sampling scheme improve the statistical accuracy of system identification in settings where periodic excitation is impractical.
\end{abstract}


\end{frontmatter}

\section{Introduction}
Event-triggered sampling, under which nonuniform sampling naturally arises, can improve resource usage and performance in various control and estimation tasks \citep{astrom2002comparison,miskowicz2007asymptotic}.
However, most system identification theory and algorithms have been developed for discrete-time models assuming uniform sampling \citep{lennart1999system}, largely because such models align directly with digital controllers.
In contrast, continuous-time models remain valid under arbitrary sampling schemes and preserve physical interpretability \citep{rao2006identification,garnier2008direct,gonzalez2025identification}.
As nonuniform sampling becomes increasingly prevalent in modern data-acquisition and control architectures, continuous-time system identification methods capable of handling irregular sampling are essential \citep{mu2015identification}.

The informativeness of the data, meaning the extent to which the collected input–output measurements allow the system to be accurately identified, depends critically on both the excitation signal and the sampling schedule. 
Multisine inputs are appealing due to their sparse spectral content and clear frequency-domain interpretation \citep{schoukens1994advantages,gonzalez2021consistent}, properties that remain intact under nonuniform sampling and make them well suited for identifying continuous-time models from irregular measurements.
However, existing multisine-based identification methods typically rely on periodic excitation and uniform sampling \citep{morelli2003multiple,grauer2020aircraft}. 
In practice, achieving periodicity is often impractical, since the fundamental period ensuring that all frequency components complete an integer number of cycles may be prohibitively large.
Nonperiodic excitation leads to spectral leakage and correlated frequency components \citep{schoukens2006analysis}, while irregular sampling undermines the assumptions behind standard frequency-domain estimators.
Consequently, the design of identification experiments is typically constrained to periodic excitation and uniform sampling.

Classical experiment design theory \citep{pukelsheim2006optimal,pronzato2008optimal} and its extensions to system identification \citep{lennart1999system,bombois2011optimal,hjalmarsson2005experiment} provide systematic tools for shaping excitation signals to maximize informativeness of data.
However, these formulations typically assume fixed sampling instants, most often uniformly spaced, and therefore focus exclusively on input design while leaving the sampling schedule outside the design loop. 
As a consequence, they are well suited for periodic or near-periodic excitations but become insufficient when the excitation is nonperiodic.
In that case, the informativeness of the data can no longer be controlled by input design alone, motivating the introduction of a complementary viewpoint in which the sampling process itself is treated as a design variable to compensate for the loss of periodicity.  

The main contribution of this work is two-fold. 
First, we derive a nonparametric frequency response function (FRF) estimator at the excited frequencies suitable for nonperiodic multisine excitation and arbitrary sampling schemes.
Within a probabilistic framework where sampling instants are modeled as random variables drawn from a tunable distribution, we show that the FRF can be estimated directly using a least-squares procedure and that the estimator is almost surely well-posed for any continuous sampling distribution. 
Furthermore, we establish unbiasedness conditions and derive approximate high-probability covariance bounds that quantify how estimation accuracy depends on the chosen sampling distribution.
Second, we design a sampling strategy that aims to improve estimation accuracy by reducing spectral leakage.
The previous results motivate treating the sampling distribution as a design variable: since well-posedness is guaranteed for any continuous sampling distribution, we optimize over families of sampling patterns rather than individual deterministic sequences, thereby extending classical experiment-design principles to continuous-time system identification with random sampling.  
To exploit this design freedom, we adopt a D-optimality criterion, maximizing the determinant of the expected Fisher Information Matrix (FIM), and show that the resulting sampling strategies enhance informativeness and reduce spectral leakage. 
Finally, we propose an iterative algorithm to construct deterministic irregular sampling sequences that retain the desirable properties predicted by the theoretical analysis while enforcing bounds on the inter-sample time.
The resulting deterministic sampling schemes are thus readily implementable in practical applications.

The remainder of the paper is organized as follows. 
Section~\ref{sec:problem_formulation} introduces the problem formulation.
Section~\ref{sec:non_param_estimation_frf} presents the nonparametric FRF estimator and analyzes its statistical properties. 
Section~\ref{sec:sampling_design_d_optimal} develops the D-optimal sampling design framework. 
Section~\ref{sec:simulation_results} reports Monte Carlo simulation results that compare the proposed methodology to uniform sampling. 
Finally, Section~\ref{sec:conclusion} summarizes the contributions.

\textit{Notation.}
The sets $\mathbb{Z}$, $\mathbb{R}$, $\mathbb{C}$ denote the sets of all integer, real and complex numbers, respectively.
Vectors are considered as column vectors, unless transposed.
Given a vector $v\in\mathbb{R}^n$ (matrix $A\in\mathbb{R}^{n\times m}$), then $x^{\top}$ ($A^{\top})$ represents the vector (matrix) transpose.
The Euclidean norm of the vector $v$ is denoted by \( ||v||_2 \) and the spectral norm of the matrix $A$ is denoted by $||A||_2$.
The notation $A=\textnormal{diag}(A_i) \in \mathbb{R}^{n \times m}$ defines a block diagonal matrix composed of the individual matrices $A_i \in \mathbb{R}^{n_i \times m_i}$ placed along the diagonal, with zeros in the off-diagonal blocks.
The minimum and maximum eigenvalue of $A$ are denoted by $\lambda_{\textnormal{min}}(A)$ and $\lambda_{\textnormal{max}}(A)$, respectively.
The identity matrix of dimension $n\times n$ is indicated by $I_n$.
Operator $\mathrm{p}$ denotes the differentiation operator \(\mathrm{p} := \frac{\text{d}}{\text{d}t}\). The complex conjugate and complex conjugate transpose of a vector $x$ are denoted as $x^*$ and $x^{\hop}$, respectively. 
A property is said to hold almost surely (a.s.) if it holds with probability one.

\section{Problem Formulation}
\label{sec:problem_formulation}
We consider the identification of a continuous-time, single-input–single-output (SISO) linear time-invariant (LTI) system from nonuniformly sampled input–output data. The system dynamics is described by
\begin{equation}
    {x}(t) = G(\mathrm{p}) u(t)
    \label{eq:system_dynamics},
\end{equation}
where $u(t) \in \mathbb{R}$ is the input signal, $x(t) \in \mathbb{R}$ is the noise-free output signal and \(G(\mathrm{p})\) represents the unknown continuous-time operator.
The system is excited by a nonperiodic multisine signal
\begin{equation}
u(t) = a_0+\sum_{m=1}^{M} a_m \cos(\omega_m t + \phi_m),
\label{eq:multisine}
\end{equation}
where \(\{\omega_m\}_{m=1}^M\) are the chosen nonzero excitation frequencies, and \(a_m\), \(\phi_m\) are the corresponding amplitudes and phase shifts.
Since each sinusoid $a_m \cos(\omega_m t + \phi_m)$ excites the system at the pair of frequencies $\pm \omega_m$, the resulting set of excited frequencies is \(\Omega:= \{0\}\cup\{\pm \omega_m\}_{m=1}^M\).

\begin{assumption}
\label{assump:nonzero_amplitudes}
The amplitudes $a_m$ of the multisine input \eqref{eq:multisine} are different from zero for every $m=0,\dots, M$.
\end{assumption}
\begin{assumption}
\label{assump:input_steady_state}
The input signal $u(t)$ can be evaluated at any time instant. 
Measurements are collected after the transient response of the system has decayed, so that $x(t)$ corresponds to the steady-state response of the system.
\end{assumption}

The output is measured at sampling instants $\{t_k\}_{k=1}^N$ and the measured output is given by 
\begin{equation}
    y(t_k)=x(t_k)+e(t_k), \quad k=1,\dots, N,
    \label{eq:noisy_measurement_output}
\end{equation}
with $e(t_k) \in \mathbb{R}$ representing additive measurement noise at time $t_k$.  

\begin{assumption}
\label{assump:noise}
The measurement noise $\{e(t_k)\}_{k=1}^N$ in \eqref{eq:noisy_measurement_output} is independent and identically distributed (i.i.d.), zero-mean with variance $\sigma^2$, and independent of both the input $\{u(t_k)\}_{k=1}^N$ and the sampling instants $\{t_k\}_{k=1}^N$.
\end{assumption}

Given the irregularly sampled dataset \(\{u(t_k),y(t_k)\}_{k=1}^N\), the goal is to estimate the FRF of the system \(G(\mathrm{p})\) at the excited frequencies \(\Omega\) under arbitrary (nonuniform) sampling instants.
In addition, we aim to construct sampling schemes that improve the statistical accuracy of the resulting FRF estimator and are practically implementable.

\section{Non-parametric Estimation of the Frequency Response Functions}
\label{sec:non_param_estimation_frf}

To analyze the impact of irregular sampling in a general setting, we model the sampling instants as random variables drawn independently from a continuous probability density $p_T(t)$ supported on the finite interval $[0,T]$, i.e.,
\begin{equation}
t_k \sim p_T(t), \quad k=1,\dots,N.
\label{eq:sampling_distribution}
\end{equation}
This probabilistic formulation allows us to capture entire classes of sampling patterns in a unified way. 
As a result, the sampling distribution $p_T(t)$ can be treated as a design variable within the optimal experiment design framework developed in Section \ref{sec:sampling_design_d_optimal}.
While the probabilistic formulation is used as an analytical and design tool to characterize and optimize sampling performance in expectation, the final outcome of the proposed approach is a deterministic sampling scheme that can be implemented in practice.

We begin by estimating the FRF of the system at the excited frequencies.
At steady state, the response of the system \eqref{eq:system_dynamics} to the multisine input \eqref{eq:multisine} is a finite superposition of sinusoids at the excited frequencies:
\begin{align*}
\hspace{-0.01cm}
x(t)
  &= G(0)a_0
   + \frac{1}{2}\sum_{m=1}^M |G(j\omega_m)| a_m \bigl(
      e^{ j(\omega_m t + \phi_m + \angle G(j\omega_m)) } \nonumber\\
  &\qquad\qquad\qquad\qquad
      +\, e^{-j(\omega_m t + \phi_m + \angle G(j\omega_m)) }
     \bigr).
\label{eq:time_response_multisine_exponential}
\end{align*}
The steady-state output is a linear combination of complex exponentials at the excited frequencies \(\omega \in \Omega\),
each scaled in magnitude and shifted in phase by the corresponding FRF value.
Thanks to the discrete spectrum of the multisine excitation, the finite set of frequency response values \(\{G( j\omega)\}_{\omega \in \Omega}\) completely describes the steady-state response of the system to the multisine excitation.

Collecting these unknown FRF values into the parameter vector $\theta \in \mathbb{C}^{2M+1}$
\begin{equation*}
    \theta = 
    \begin{bmatrix}
        G(0),\,
        G(j\omega_1),\, G(-j\omega_1),\,
        \dots,\,
        G(j\omega_M),\, G(-j\omega_M)
    \end{bmatrix}^{\top},
    \label{eq:parameter_vector}
\end{equation*}
the measured output at time \(t_k\) satisfies the linear regression model $y(t_k) = \varphi^{\top}(t_k)\,\theta + e(t_k), \; k=1,\dots,N$, where the regressor vector collects the known contributions of each input frequency to the output at time $t_k$, i.e.,
\begin{equation*}
    \varphi(t_k) =
    \begin{bmatrix}
        a_0 \\[2pt]
        \tfrac{a_1}{2} e^{j(\omega_1 t_k+\phi_1)} \\
        \tfrac{a_1}{2} e^{-j(\omega_1 t_k+\phi_1)} \\
        \vdots \\
        \tfrac{a_M}{2} e^{j(\omega_M t_k+\phi_M)} \\
        \tfrac{a_M}{2} e^{-j(\omega_M t_k+\phi_M)}
    \end{bmatrix}
    \in \mathbb{C}^{2M+1}.
\end{equation*}

Stacking all $N$ measurements yields
\begin{equation}
    Y_N = \Phi_N\,\theta + E_N,
    \label{eq:linear_regression}
\end{equation}
where \(Y_N\in\mathbb{R}^N\), \(E_N\in\mathbb{R}^N\), and \(\Phi_N\in\mathbb{C}^{N\times (2M+1)}\).
To separate the sampling-dependent terms from the input parameters, we factorize
\begin{equation}
    \Phi_N = \Psi_N\,A,
    \label{eq:factorization_Phi}
\end{equation}
where 
\begin{equation*}
    A = \operatorname{diag}(a_0, \tfrac{a_1}{2}e^{j\phi_1},  \tfrac{a_1}{2}e^{-j\phi_1},  \dots, \allowbreak \tfrac{a_M}{2}e^{j\phi_M},  \tfrac{a_M}{2}e^{-j\phi_M})
\end{equation*}
contains the input amplitudes and phases, and \(\Psi_N\) is the generalized Vandermonde matrix \citep{lundengaard2017generalized}
\begin{equation}
    \Psi_N =
    \begin{bmatrix}
    1 & e^{j\omega_1 t_1} & e^{-j\omega_1 t_1} & \cdots & e^{j\omega_M t_1} & e^{-j\omega_M t_1} \\
    \vdots & \vdots & \vdots & \ddots & \vdots & \vdots \\
    1 & e^{j\omega_1 t_N} & e^{-j\omega_1 t_N} & \cdots & e^{j\omega_M t_N} & e^{-j\omega_M t_N}
    \end{bmatrix}.
    \label{eq:Psi_N}
\end{equation}
The least-squares estimator of the FRF values is then
\begin{equation}
    \hat{\theta}_{LS}
    = \argmin\limits_{\theta} \|Y_N - \Phi_N \theta\|_2^2
    = (\Phi_N^\hop \Phi_N)^{-1} \Phi_N^\hop Y_N.
    \label{eq:LS_solution}
\end{equation}
We have obtained a nonparametric least-squares estimator of the FRF at the excited frequencies under arbitrary sampling.
Importantly, the estimator depends explicitly on the sampling instants through \(\Psi_N\), which allows us to analyze its statistical properties as a function of $p_T(t)$. 
In particular, the expected information content and the variance of \(\hat{\theta}_{LS}\) can be directly related to the choice of sampling instants.  

\subsection{Statistical Properties of the Least-Squares Estimator}
\label{sec:properties_lse}
Understanding how the sampling distribution affects the statistical properties of the estimator is key to the contributions of this work.
In this section, we analyze well-posedness and accuracy of \(\hat{\theta}_{LS}\) under random sampling.
These results provide the foundation for designing D-optimal sampling schemes that reduce spectral leakage and maximize data informativeness under nonperiodic input.

Let $T_N = [t_1,\dots,t_N]^\top$ denote the vector of random sampling instants and let $\Phi_N=\Phi_N(T_N)$ be the associated regression matrix.
The estimator $\hat{\theta}_{LS}$ is well-posed when the FIM $\Phi_N^\hop \Phi_N$ is nonsingular.
Under uniform sampling, the well-posedness of the least-squares problem depends on the input excitation and on the sampling period \citep{gonzalez2024sampling}. 
In contrast, random sampling provides flexibility in the placement of measurements, which can be exploited in case of nonperiodic excitation.
Degenerate configurations that render $\Phi_N^\hop \Phi_N$ singular occur only under specific algebraic relations that are extremely unlikely under a random sampling scheme.
Lemma \ref{lem:well_posedness_ls_estimator} formalizes this intuition by showing that these degenerate cases occur with probability zero.

\vspace{0.5em}
\begin{lemma}[Well-posedness of the least-squares estimator]
\label{lem:well_posedness_ls_estimator}
Under Assumption \ref{assump:nonzero_amplitudes}, the least-squares estimator \eqref{eq:LS_solution} for the linear regression model \eqref{eq:linear_regression}  is well-posed almost surely provided that the number of measurements $N$ is greater than $2M$, and the sampling instants $\{t_k\}_{k=1}^N$ are drawn independently from the continuous distribution in \eqref{eq:sampling_distribution}.
\end{lemma}
\begin{proof}
To ensure that the estimator is well-posed almost surely, \(\Phi_N^{\hop} \Phi_N \in \mathbb{C}^{(2M+1)\times(2M+1)}\) must be nonsingular almost surely, i.e., \(\Phi_N \in \mathbb{C}^{N \times (2M+1)}\) must have full column rank almost surely.
We consider the factorization in \eqref{eq:factorization_Phi}.
By Assumption \ref{assump:nonzero_amplitudes}, $A$ is nonsingular. Therefore, it suffices to show that $\Psi_N$ has full column rank almost surely.
Let $C = 2M+1$ denote the number of columns of $\Psi_N$.
Since $\Psi_N \in \mathbb{C}^{N\times C}$ is a rectangular matrix with $N\geq C$, it has full column rank almost surely if and only if there exists with probability one at least one $C\times C$ submatrix of $\Psi_N$ with nonzero determinant.
Let $\mathcal{I} = \{i_1, \dots, i_C\} \subset \{1, \dots, N\}$ denote any subset of $C$ distinct row indices and consider the corresponding square submatrix $\Psi_{\mathcal{I}}$.
Its determinant $f_{\mathcal{I}}(t_{i_1},\dots,t_{i_C})=\textnormal{det}(\Psi_{\mathcal{I}})$ is a complex-valued analytic function of the sampling instants $(t_{i_1},\dots,t_{i_C})$.
By the structure of \eqref{eq:Psi_N}, $f_{\mathcal{I}}(t_{i_1},\dots,t_{i_C})$ is not identically zero for all $(t_{i_1},\dots,t_{i_C}) \in[0,T]$.  
Since $f_{\mathcal{I}}$ is analytic and not identically zero, its zero set
\[
Z_{\mathcal{I}} = \{ (t_{i_1},\dots,t_{i_C}) \in [0,T] : f_{\mathcal{I}}(t_{i_1},\dots,t_{i_C}) = 0 \}
\]
has Lebesgue measure zero in $[0,T]$ \citep[Sec. 7.3]{krantz2001function}. 
Because the sampling instants are independently drawn from a continuous distribution, the probability that $(t_{i_1},\dots,t_{i_C})$ lies in $Z_{\mathcal{I}}$ is zero for any fixed $\mathcal{I}$.
This implies that \(\Psi_N\) has full column rank almost surely and thus $\Phi_N^\hop \Phi_N$ is nonsingular almost surely.
\end{proof}

With well-posedness guaranteed under any continuous sampling distribution, we are now justified in analyzing the statistical properties of the estimator and in treating the sampling distribution itself as a design variable.
We start by investigating how the choice of sampling distribution affects bias and covariance.

\vspace{0.5em}
\begin{lemma}[Unbiasedness of the least-squares estimator]
\label{lem:unbiased_random}
Consider the least-squares estimator \eqref{eq:LS_solution} for the linear regression model \eqref{eq:linear_regression}.
Suppose the sampling instants $\{t_k\}_{k=1}^N$ are drawn independently from the continuous distribution in \eqref{eq:sampling_distribution}, and that Assumptions~\ref{assump:nonzero_amplitudes} to \ref{assump:noise} hold.
If $N>2M$, then $\hat{\theta}_{LS}$ is an unbiased estimator of the FRF at the excitation frequencies $\omega \in \Omega$, i.e.,
$\mathbb{E}\{\hat{\theta}_{LS}\} = \theta$.
\end{lemma}

\begin{proof}
By Lemma~\ref{lem:well_posedness_ls_estimator}, the matrix $\Phi_N^\hop\Phi_N$ is invertible with probability one, so the estimator is almost surely given by
\[
\hat{\theta}_{LS}
= (\Phi_N^\hop\Phi_N)^{-1}\Phi_N^\hop Y_N
= \theta + (\Phi_N^\hop\Phi_N)^{-1}\Phi_N^\hop E_N.
\]
Conditioning on a realization of the random sampling vector $T_N$, \eqref{eq:linear_regression} becomes deterministic in $\Phi_N$.
Taking the conditional expectation with respect to the noise and using Assumption~\ref{assump:noise}, we obtain
\[
\mathbb{E}\{\hat{\theta}_{LS} \mid T_N\} = \theta + (\Phi_N^\hop \Phi_N)^{-1} \Phi_N^\hop \, \mathbb{E}\{E_N \mid T_N\}=\theta.
\]
Finally, applying the law of iterated expectations yields
\[
\mathbb{E}\{\hat{\theta}_{LS}\} = \mathbb{E}_{T_N}\{\mathbb{E}\{\hat{\theta}_{LS} \mid T_N\}\} = \mathbb{E}_{T_N}\{\theta\} = \theta,
\]
proving that the estimator is unbiased.
\end{proof}

We now analyze the covariance of the estimate under random sampling.
Conditioned on a realization of the random sampling vector $T_N$, the estimation error can be written as $\hat{\theta}_{LS}-\theta = (\Phi_N^\hop\Phi_N)^{-1}\Phi_N^\hop E_N$.
Using Assumption \ref{assump:noise} and the fact that $\mathbb{E}\{\hat{\theta}_{LS}\mid T_N\}=\theta$, the law of total covariance gives
\begin{equation*}
\hspace{-0.01cm}
    \operatorname{Cov}\{\hat{\theta}_{LS}\}
= \mathbb{E}_{T_N}\big[\operatorname{Cov}\{\hat{\theta}_{LS}\mid T_N\}\big]
= \sigma^2 \, \mathbb{E}_{T_N} \big[(\Phi_N^\hop\Phi_N)^{-1}\big].
\label{eq:covariance_Phi}
\end{equation*}
Exploiting the factorization in \eqref{eq:factorization_Phi}, we can separate the effects of the input and the sampling pattern as
\begin{equation}
\operatorname{Cov}\{\hat{\theta}_{LS}\} = \sigma^2 \, A^{-1} \, \mathbb{E}_{T_N}\big[(\Psi_N^\hop\Psi_N)^{-1}\big] \, A^{-\hop}.
\label{eq:covariance_psi}
\end{equation}
Computing the expectation analytically is generally intractable due to the nonlinear dependence on the sampling instants.
To obtain a finite-sample characterization of the estimator uncertainty, we first derive a high-probability bound on \((\Psi_N^\hop\Psi_N)^{-1}\), which serves as an intermediate result for an approximate bound on the covariance.

\vspace{0.5em}
\begin{lemma}[High-probability bound]
\label{lem:high_prob_bound_cov}
    Consider $R_N = \Psi_N^\hop \Psi_N$ and $R_0=\mathbb{E}[R_N]$.
    Suppose that the sampling instants $\{t_k\}_{k=1}^N$ are drawn independently from the continuous distribution in \eqref{eq:sampling_distribution}. Choose $t \in (0, \lambda_{min}(R_0))$ such that
    \begin{equation*}
        (2M+1)\exp\!\left(-\frac{t^2/2}{v + Rt/3}\right) = \delta,
        \label{eq:t_delta_relation}
    \end{equation*}
    where $R \le 2(2M+1)$ and $v \le 4N(2M+1)^2$.
    Then, with probability at least $1 - \delta$, $R_N$ satisfies
    \begin{equation}
        R_N^{-1}
        \preceq \frac{1}{\lambda_{\min}(R_0) - t(\delta)}I_{2M+1}.
        \label{eq:Rn_bound}
    \end{equation}
\end{lemma}
\begin{proof}
    Reported in Appendix \ref{app:proof_lemma_cov}.
\end{proof}
To translate the high-probability bound on $R_N^{-1}$ into a bound on the covariance, 
let $\mathcal{G}$ denote the event $\mathcal{G} = \{\|R_N - R_0\| \le t\}$ and $\mathcal{G}^c$ its complement. 
From \eqref{eq:matrix_bernstein_explicit}, $\Pr(\mathcal{G}) \ge 1 - \delta$. 
By the law of total expectation,
\[
\mathbb{E}[R_N^{-1}] 
= \mathbb{E}[R_N^{-1} \mid \mathcal{G}]\,\Pr(\mathcal{G})
+ \mathbb{E}[R_N^{-1} \mid \mathcal{G}^c]\,\Pr(\mathcal{G}^c).
\]
On $\mathcal{G}$, the exact bound \eqref{eq:Rn_bound} holds. 
On $\mathcal{G}^c$, $\|R_N^{-1}\|$ is finite almost surely since $R_N$ is positive definite almost surely due to Lemma~\ref{lem:well_posedness_ls_estimator}. 
Hence, for small $\delta$, the tail contribution can be neglected, yielding the approximation
\begin{align}
\operatorname{Cov}\{\hat{\theta}_{LS}\} 
& \approx \sigma^2 \, A^{-1} \, (\mathbb{E}[R_N^{-1}|\mathcal{G}]) A^{-\hop} \nonumber \\
& \preceq \sigma^2 \, A^{-1} \, \frac{1}{\lambda_{\min}(R_0) - t} A^{-\hop}
,
\label{eq:covariance_bound}
\end{align}
which relies on the event $\mathcal{G}$ occurring with probability $1-\delta$ and becomes more accurate as $\delta \to 0$.

Although the bound \eqref{eq:covariance_bound} is an approximation that neglects the contribution of rare tail events, it still provides useful insight into how the sampling distribution affects estimation accuracy.
In particular, \eqref{eq:covariance_bound} depends explicitly on $R_0 = \mathbb{E}[\Psi_N^\hop \Psi_N]$, whose diagonal entries are equal to $N$, regardless of the sampling scheme.
The off-diagonal terms of $R_0$ capture the correlation between regressors, which degrades numerical conditioning and estimation accuracy.
For two distinct excitation frequencies $\omega_p$ and $\omega_q$ at sampling time $t_k$, the corresponding off-diagonal term is
\begin{align}
    [R_0(p_T)]_{pq} = N \int_0^T e^{j(\omega_p - \omega_q) t} p_T(t)\, dt.
\label{eq:diagonal_term_R0}
\end{align}

These off-diagonal terms directly affect the expected FIM and, consequently, the expected informativeness of the identification experiments. 
Under uniform sampling, the FIM is diagonal for periodic multisine excitation, i.e., when all excited frequencies are integer multiples of $2\pi/T$, because the regressors are orthogonal over the observation window.
For nonperiodic excitation, uniform sampling produces a non-diagonal FIM due to spectral leakage \citep{schoukens2006analysis, pintelon2012system, gonzalez2025finite}.
To mitigate this effect under random sampling, we can select a sampling distribution $p_T(t)$ that suppresses the off-diagonal terms of $R_0$.
Since $\mathbb{E}[\Phi_N^\hop\Phi_N]=A^\hop R_0A$, this gives an expected FIM that is diagonal.
In other words, an optimized design for $p_T(t)$ can reduce spectral leakage in expectation, with the expected informativeness of the experiment directly controlled through $R_0$.
When $p_T(t)$ is designed in this way, $R_0 = N I_{2M+1}$ and the covariance bound \eqref{eq:covariance_bound} simplifies to 
\begin{equation}
\operatorname{Cov}\{\hat{\theta}_{LS}\} 
\preceq A^{-1}  \frac{\sigma^2}{N - t(\delta)} \, A^{-\hop},
\end{equation}
which approximate the classical $1/N$ variance decay \citep{lennart1999system}.

\section{Sampling Design via D-Optimality}
\label{sec:sampling_design_d_optimal}
Maximizing the determinant of the FIM $\mathbb{E}[\Phi_N^\hop \Phi_N] = A^\hop R_0 A$ corresponds to maximizing the expected information content of the identification experiment.
Since the input matrix $A$ is fixed, this is equivalent to maximize $\det(R_0)$. 
Since $R_0$ is Hermitian with fixed diagonal entries, Hadamard’s inequality \citep[Th.~7.8.1]{horn2012matrix} implies that $\det(R_0)$ is maximized when the off-diagonal entries are minimized. 
These off-diagonal terms \eqref{eq:diagonal_term_R0} correspond to spectral leakage, so suppressing leakage is equivalent to maximize data informativeness. 
This leads naturally to a D-optimal design criterion, where $p_T(t)$ is chosen to maximize $\det(R_0(p_T))$.

Formally, we rewrite
\begin{equation}
    R_0(p_T) = N \int_0^T \psi(t)\psi^\hop(t) \, p_T(t)\, dt,
    \label{eq:R0_pT}
\end{equation}
where 
\begin{equation}
    \psi(t) = 
    \begin{bmatrix} 
    1 & e^{-j\omega_1 t} & e^{j\omega_1 t} & \dots & e^{-j\omega_M t} & e^{j\omega_M t} 
    \end{bmatrix}^\top.
    \label{eq:psi_tk}
    \nonumber
\end{equation}
Since minimizing $-\log \det(\cdot)$ is equivalent to maximizing $\det(\cdot)$ and is more convenient for analysis, the resulting D-optimal sampling design problem
\begin{align}
\label{eq:D_opt_problem_precise}
\min_{p_T} \quad & - \log \det \big( R_0(p_T) \big) \\
\text{s.t.} \quad 
& p_T(t) \ge 0 \quad\text{for } t\in[0,T], \nonumber\\
& \int_0^T p_T(t)\, dt = 1. \nonumber
\end{align}
After a standard convexity analysis, we note that the feasible set of probability densities on $[0,T]$ is convex and that $-\log\det(R_0(p_T))$ is convex on this set.  
Hence, \eqref{eq:D_opt_problem_precise} is a convex optimization problem and admits a solution \citep{boyd2004convex}.  
In summary, the D-optimal sampling distribution maximizes the information content of the experiment and, by construction, mitigates spectral leakage, leading to improved estimation accuracy.

\subsection{From Optimal Density to Deterministic Sampling Instants}

Our ultimate goal is not the optimal probability density, but a deterministic sequence of sampling instants implementable in practice.  
To this end, we consider a discretization of the time interval $[0,T]$ with a fine grid $\{t^{(1)},\dots,t^{(L)}\}$ and we employ an iterative greedy algorithm.
Let $R_k$ denote the information matrix after $k$ samples, i.e., $R_k = \sum_{i=1}^k \psi(t_i)\psi^\hop(t_i)$.
Adding a candidate point $t$ gives the rank-one update $R_{k+1} = R_k + \psi(t)\psi^\hop(t)$, with the corresponding log-determinant increase
\begin{equation*}
\begin{aligned}
    \Delta_k(t) &= \log\det(R_{k+1}) - \log\det(R_k) 
    = \log(1 + s_k(t)),
\end{aligned}
\label{eq:log_det_increase}
\end{equation*}
where $s_k(t) := \psi^\hop(t) R_k^{-1}\psi(t)$. 
By selecting the next sampling instant as
\[
t_{k+1} \in \argmax_{t \in \{t^{(1)},\dots,t^{(L)}\}} s_k(t),
\]
each new point is chosen as the one that provides the largest incremental contribution to the information content of the identification experiment.
This procedure provides a practical way to approximate the structural properties of the D-optimal sampling distribution derived from the probabilistic formulation, while yielding a deterministic and implementable sampling sequence. 
In particular, the grid resolution $T/L$ controls how closely the discrete design matches the optimal density: finer grids enable better approximation of the desired sampling scheme, whereas coarser grids enforce larger inter-sample time.
Consequently, the proposed approach approximates the information-theoretic optimal design while explicitly bounding inter-sample time.

\section{Simulation Results}
\label{sec:simulation_results}

To validate the proposed sampling strategy in a system identification setting, we consider the second-order system
\begin{equation}
    G(\mathrm{p}) = \frac{2}{\mathrm{p}^2 + 2 \mathrm{p} + 1},
\end{equation}
excited with a fixed nonperiodic multisine with frequencies \(\Omega=\{0, \pm 3\sqrt{2}, \allowbreak \pm 3\pi, \pm 12, \pm 12.3\}\) over an observation window of length \(T = 5\,\mathrm{s}\). 
The input is kept fixed across all experiments so that differences in identification performance can be attributed solely to the sampling strategy.
We vary the number of measurements $N$ taken in the interval and compare two sampling strategies: uniform sampling and the D-optimal deterministic sequence designed as detailed in Section~\ref{sec:sampling_design_d_optimal}.
For each $N$, we construct the least-squares FRF estimator described in Section~\ref{sec:non_param_estimation_frf} and we perform \(500\) Monte Carlo simulations. 

Figure~\ref{fig:accuracy_opt_vs_uni} shows the sum of the variances of the estimated parameters $\mathrm{tr}(\operatorname{Cov}\{\hat \theta_{LS}\})$ for each value of $N$.  
The optimized sampling sequence consistently yields lower variance than uniform sampling, demonstrating improved estimation accuracy.
Notably, the results indicate that, for a given level of estimation accuracy, the optimized sampling scheme requires roughly half the number of samples compared to uniform sampling.  
Figure~\ref{fig:metrics_uni_vs_opt} compares the D-optimal and uniform sampling strategies in terms of the average information density $\det(\text{FIM}/N)$ over the Monte Carlo runs. 
The D-optimal sampling strategy consistently achieves higher information density.
This indicates that, on average, each measurement contributes more Fisher information under the D-optimal design, resulting in more informative and reliable parameter estimates with the same number of samples.
These results highlight the practical advantage of the D-optimal sampling design: it allows for a more efficient experiment by reducing measurement effort while maintaining or improving estimation quality.

\begin{figure}
    \centering
    \includegraphics[width=0.9\linewidth]{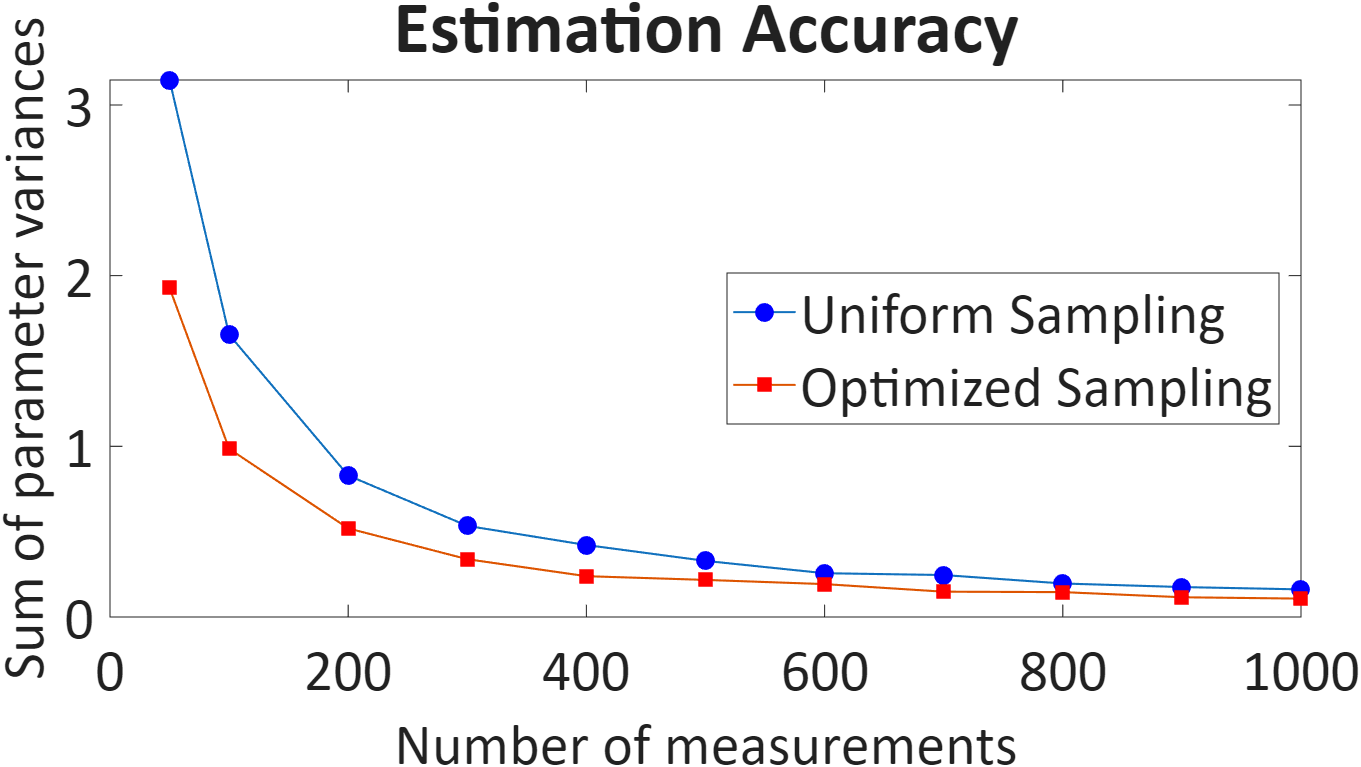}
    \caption{Estimation accuracy of the least-squares FRF estimator for uniform sampling and the D-optimal sequence.}

    \label{fig:accuracy_opt_vs_uni}
\end{figure}

\begin{figure}
    \centering
    \includegraphics[width=0.9\linewidth]{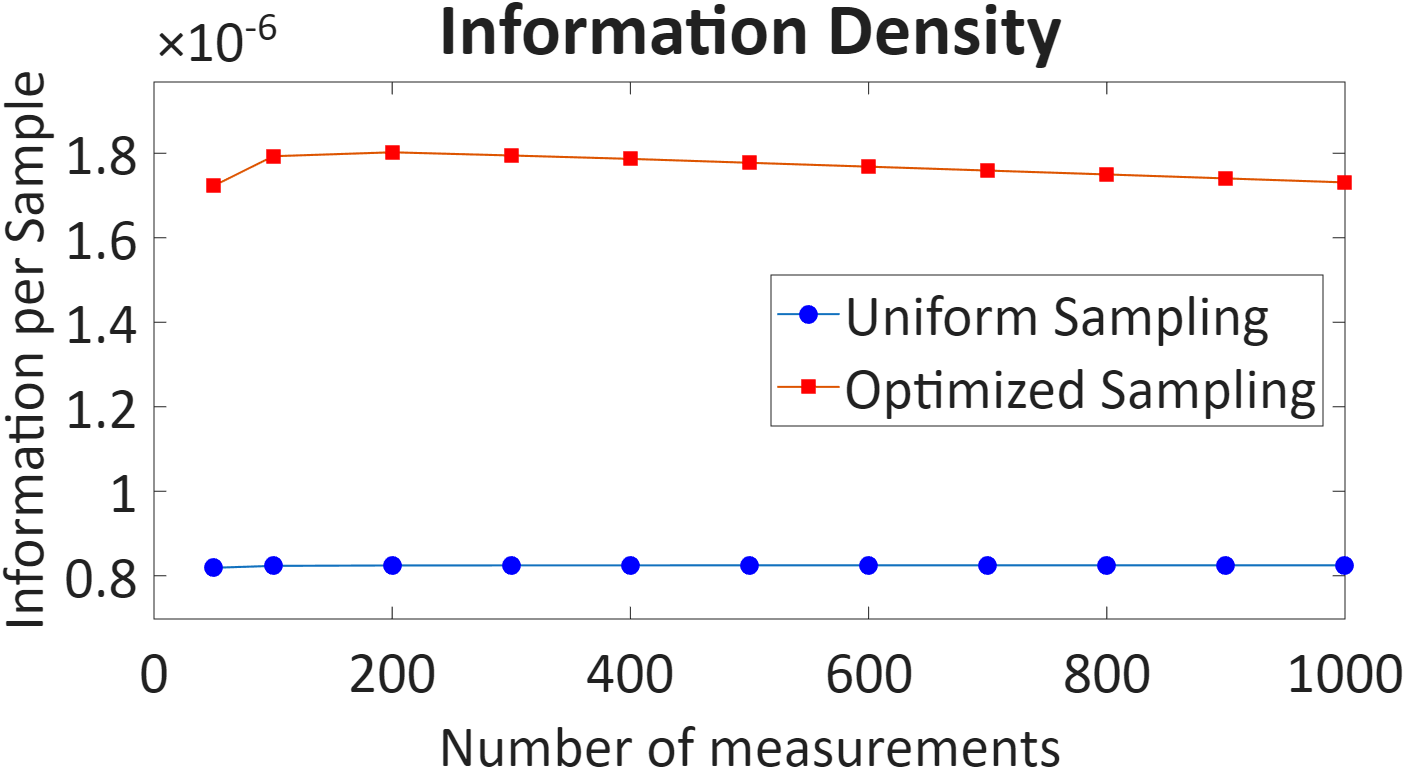}
    \caption{Information density of the least-squares FRF estimator for uniform sampling and the D-optimal sequence.}
    \label{fig:metrics_uni_vs_opt}
\end{figure}

\section{Conclusion}
\label{sec:conclusion}
This work has presented a framework for designing nonuniform sampling strategies in continuous-time system identification under nonperiodic multisine excitation.  
By treating the sampling distribution as a design variable and optimizing it via a D-optimal criterion, we extend classical experiment design principles to system identification with random sampling.  
The resulting theory highlights the connection between information maximization and the suppression of spectral leakage, and is complemented by a practical iterative algorithm for constructing deterministic sampling sequences that closely approximate the optimal design.  
Simulations on a second-order system illustrate that D-optimal sampling sequences yield more informative measurements and more accurate FRF estimates than uniform sampling.  
These results demonstrate that sampling design provides a powerful means to improve the statistical efficiency of continuous-time system identification when periodic excitation cannot be enforced.

\appendix
\section{Proof Lemma \ref{lem:high_prob_bound_cov}}
\begin{proof}
\label{app:proof_lemma_cov}
We rewrite $R_N = \Psi_N^\hop \Psi_N = \sum_{k=1}^N \psi_k \psi_k^\hop$, where 
\[ 
\psi_k = 
\begin{bmatrix} 
e^{j\beta_0 t_k} & e^{j\beta_1 t_k} &e^{j\beta_2 t_k} & \dots & e^{j\beta_{2M-1} t_k} & e^{j\beta_{2M} t_k} 
\end{bmatrix}^\top, 
\] 
and $\beta_0=0, \beta_1=-\omega_1, \beta_2=\omega_1, \dots, \beta_{2M-1}=-\omega_M, \beta_{2M}=\omega_M \allowbreak$.
By Lemma \ref{lem:well_posedness_ls_estimator}, the matrix $\Psi_N$ has full column rank almost surely.  
Consequently, $R_N$ is Hermitian positive definite almost surely. 
The vectors $\psi_k$ are independent random vectors with bounded Euclidean norm: 
\begin{equation*}
    \|\psi_k\psi_k^\hop\|=\|\psi_k\|_2^2 = \sum_{\ell=1}^{2M+1}|e^{j\beta_{\ell}t_k}|^2=2M+1.
\end{equation*}
We define the zero-mean matrices $X_k = \psi_k \psi_k^{\hop} - \mathbb{E}[\psi_k \psi_k^{\hop}]$, so that $R_N - R_0 = \sum_{k=1}^{N} X_k$.
The matrices $X_k$ are independent, Hermitian, and, using the triangle inequality for the spectral norm, they satisfy 
$$
R \coloneqq \|X_k\| \le \|\psi_k \psi_k^\hop\| + \|\mathbb{E}[\psi_k \psi_k^\hop]\| \le 2(2M+1).
$$
From the submultiplicativity of the spectral norm,
$
\|X_k^2\| \le \|X_k\|^2 \le 4(2M+1)^2.
$
Since \(X_k\) are independent,
\[
v \coloneqq \left\|\sum_{k=1}^N \mathbb{E}[X_k^2]\right\| \le \sum_{k=1}^N \|\mathbb{E}[X_k^2]\| \le N \cdot 4(2M+1)^2.
\]

By applying the Matrix Bernstein inequality to the sum of $X_k$ \citep[Th.1.4]{tropp2012user}, we obtain a probabilistic bound on the deviation of $R_N$ from its mean:
\begin{equation*}
    \Pr\big(\|R_N - R_0\| \ge t\big)
 \le (2M+1) \exp\!\left( - \frac{t^2/2}{v + Rt/3} \right).
\label{eq:matrix_bernstein}
\end{equation*}
Substituting $R = 2(2M+1)$ and $v = 4N(2M+1)^2$,
\begin{align}
\label{eq:matrix_bernstein_explicit}
&\Pr\big(\|R_N - R_0\| \ge t\big)\\ 
&\le (2M+1) \exp\!\Bigg( - \frac{t^2/2}{4 N (2M+1)^2 + 2 (2M+1)\, t/3} \Bigg). \nonumber
\end{align}
Since $R_N - R_0$ is Hermitian, its spectral norm coincides with the largest absolute value of its eigenvalues.
Hence, $\|R_N - R_0\| < t$ is equivalent to $-t I_{2M+1} \preceq R_N - R_0 \preceq t I_{2M+1}$.
It follows that $R_0-tI_{2M+1} \preceq R_N$ holds with probability $1-\delta$.
Since $t \in (0, \lambda_{min}(R_0))$, both matrices are positive definite. Since inversion reverses the Loewner order on positive definite matrices, $R_N^{-1} \preceq (R_0 - tI_{2M+1})^{-1}$ with probability $1-\delta$, which is equivalent to
\begin{equation*}
        R_N^{-1}
        \preceq \frac{1}{\lambda_{\min}(R_0) - t(\delta)}I_{2M+1}.
        \qedhere
    \end{equation*}
\end{proof}

\begin{thebibliography}{24}
\providecommand{\natexlab}[1]{#1}
\providecommand{\url}[1]{\texttt{#1}}
\providecommand{\urlprefix}{URL }
\expandafter\ifx\csname urlstyle\endcsname\relax
  \providecommand{\doi}[1]{doi:\discretionary{}{}{}#1}\else
  \providecommand{\doi}{doi:\discretionary{}{}{}\begingroup \urlstyle{rm}\Url}\fi

\bibitem[{{\AA}str\"om and Bernhardsson(2002)}]{astrom2002comparison}
{\AA}str\"om, K.J. and Bernhardsson, B.M. (2002).
\newblock Comparison of {Riemann} and {Lebesgue} sampling for first order stochastic systems.
\newblock In \emph{Proceedings of the 41st IEEE Conference on Decision and Control, 2002}, volume~2, 2011--2016. IEEE.

\bibitem[{Bombois et~al.(2011)Bombois, Gevers, Hildebrand, and Solari}]{bombois2011optimal}
Bombois, X., Gevers, M., Hildebrand, R., and Solari, G. (2011).
\newblock Optimal experiment design for open and closed-loop system identification.
\newblock \emph{Communications in Information and Systems}, 11(3), 197--224.

\bibitem[{Boyd and Vandenberghe(2004)}]{boyd2004convex}
Boyd, S. and Vandenberghe, L. (2004).
\newblock \emph{Convex Optimization}.
\newblock Cambridge University Press.

\bibitem[{Garnier et~al.(2008)Garnier, Wang, and Young}]{garnier2008direct}
Garnier, H., Wang, L., and Young, P.C. (2008).
\newblock Direct identification of continuous-time models from sampled data: Issues, basic solutions and relevance.
\newblock In \emph{Identification of continuous-time models from sampled data}, 1--29. Springer.

\bibitem[{Gonz{\'a}lez et~al.(2025{\natexlab{a}})Gonz{\'a}lez, Classens, Rojas, Oomen, and Hjalmarsson}]{gonzalez2025finite}
Gonz{\'a}lez, R.A., Classens, K., Rojas, C.R., Oomen, T., and Hjalmarsson, H. (2025{\natexlab{a}}).
\newblock Finite sample {MIMO} system identification with multisine excitation: Nonparametric, direct, and two-step parametric estimators.
\newblock \emph{arXiv preprint arXiv:2510.26929}.

\bibitem[{Gonz{\'a}lez et~al.(2025{\natexlab{b}})Gonz{\'a}lez, Classens, Rojas, Welsh, and Oomen}]{gonzalez2025identification}
Gonz{\'a}lez, R.A., Classens, K., Rojas, C.R., Welsh, J.S., and Oomen, T. (2025{\natexlab{b}}).
\newblock Identification of additive continuous-time systems in open and closed loop.
\newblock \emph{Automatica}, 173, Art. 112013.

\bibitem[{Gonz{\'a}lez et~al.(2021)Gonz{\'a}lez, Rojas, Pan, and Welsh}]{gonzalez2021consistent}
Gonz{\'a}lez, R.A., Rojas, C.R., Pan, S., and Welsh, J.S. (2021).
\newblock Consistent identification of continuous-time systems under multisine input signal excitation.
\newblock \emph{Automatica}, 133, Art. 109859.

\bibitem[{Gonz{\'a}lez et~al.(2024)Gonz{\'a}lez, van Haren, Oomen, and Rojas}]{gonzalez2024sampling}
Gonz{\'a}lez, R.A., van Haren, M., Oomen, T., and Rojas, C.R. (2024).
\newblock Sampling in parametric and nonparametric system identification: {Aliasing}, input conditions, and consistency.
\newblock \emph{IEEE Control Systems Letters}, 8, 2415--2420.

\bibitem[{Grauer and Boucher(2020)}]{grauer2020aircraft}
Grauer, J.A. and Boucher, M.J. (2020).
\newblock Aircraft system identification from multisine inputs and frequency responses.
\newblock \emph{Journal of Guidance, Control, and Dynamics}, 43(12), 2391--2398.

\bibitem[{Hjalmarsson(2005)}]{hjalmarsson2005experiment}
Hjalmarsson, H. (2005).
\newblock From experiment design to closed-loop control.
\newblock \emph{Automatica}, 41(3), 393--438.

\bibitem[{Horn and Johnson(2012)}]{horn2012matrix}
Horn, R.A. and Johnson, C.R. (2012).
\newblock \emph{Matrix Analysis}.
\newblock Cambridge University Press.

\bibitem[{Krantz(2001)}]{krantz2001function}
Krantz, S.G. (2001).
\newblock \emph{Function theory of several complex variables}, volume 340.
\newblock American Mathematical Soc.

\bibitem[{Ljung(1999)}]{lennart1999system}
Ljung, L. (1999).
\newblock System identification: Theory for the user.
\newblock \emph{PTR Prentice Hall, Upper Saddle River, NJ}.

\bibitem[{Lundeng{\aa}rd(2017)}]{lundengaard2017generalized}
Lundeng{\aa}rd, K. (2017).
\newblock \emph{Generalized Vandermonde matrices and determinants in electromagnetic compatibility}.
\newblock Ph.D. thesis, M{\"a}lardalen University.

\bibitem[{Miskowicz(2007)}]{miskowicz2007asymptotic}
Miskowicz, M. (2007).
\newblock Asymptotic effectiveness of the event-based sampling according to the integral criterion.
\newblock \emph{Sensors}, 7(1), 16--37.

\bibitem[{Morelli(2003)}]{morelli2003multiple}
Morelli, E.A. (2003).
\newblock Multiple input design for real-time parameter estimation in the frequency domain.
\newblock \emph{IFAC Proceedings Volumes}, 36(16), 639--644.

\bibitem[{Mu et~al.(2015)Mu, Guo, Wang, Yin, Xu, and Zheng}]{mu2015identification}
Mu, B., Guo, J., Wang, L.Y., Yin, G., Xu, L., and Zheng, W.X. (2015).
\newblock Identification of linear continuous-time systems under irregular and random output sampling.
\newblock \emph{Automatica}, 60, 100--114.

\bibitem[{Pintelon and Schoukens(2012)}]{pintelon2012system}
Pintelon, R. and Schoukens, J. (2012).
\newblock \emph{System identification: A frequency domain approach}.
\newblock John Wiley \& Sons.

\bibitem[{Pronzato(2008)}]{pronzato2008optimal}
Pronzato, L. (2008).
\newblock Optimal experimental design and some related control problems.
\newblock \emph{Automatica}, 44(2), 303--325.

\bibitem[{Pukelsheim(2006)}]{pukelsheim2006optimal}
Pukelsheim, F. (2006).
\newblock \emph{Optimal design of experiments}.
\newblock SIAM.

\bibitem[{Rao and Unbehauen(2006)}]{rao2006identification}
Rao, G.P. and Unbehauen, H. (2006).
\newblock Identification of continuous-time systems.
\newblock \emph{IEE Proceedings-Control Theory and Applications}, 153(2), 185--220.

\bibitem[{Schoukens et~al.(1994)Schoukens, Pintelon, and Guillaume}]{schoukens1994advantages}
Schoukens, J., Pintelon, R., and Guillaume, P. (1994).
\newblock On the advantages of periodic excitation in system identification.
\newblock \emph{IFAC Proceedings Volumes}, 27(8), 1115--1120.

\bibitem[{Schoukens et~al.(2006)Schoukens, Rolain, and Pintelon}]{schoukens2006analysis}
Schoukens, J., Rolain, Y., and Pintelon, R. (2006).
\newblock Analysis of windowing/leakage effects in frequency response function measurements.
\newblock \emph{Automatica}, 42(1), 27--38.

\bibitem[{Tropp(2012)}]{tropp2012user}
Tropp, J.A. (2012).
\newblock User-friendly tail bounds for sums of random matrices.
\newblock \emph{Foundations of Computational Mathematics}, 12(4), 389--434.

\end{thebibliography}
\end{document}